\documentclass{ifacconf}

\usepackage{graphicx}      
\usepackage{natbib}        
\usepackage{amsmath,amssymb,amsfonts}
\usepackage[T1]{fontenc}
\usepackage[utf8]{inputenc}
\usepackage{xcolor}
\usepackage[justification=justified,singlelinecheck=false]{caption}
\usepackage{tabularx}
\usepackage{adjustbox}
\usepackage{bm}
\usepackage{listings}
\newtheorem{remark}{Remark}
\newtheorem{definition}{Definition}
\newtheorem{theorem}{Theorem}
\newtheorem{problem}{Problem}
\newtheorem{proposition}{Proposition}
\newtheorem{lemma}{Lemma}
\newtheorem{example}{Example}
\newtheorem{corollary}{Corollary}
\newtheorem{assumption}{Assumption}
\begin{document}
\begin{frontmatter}

\title{Refined Barrier Conditions for Finite-Time Safety and Reach-Avoid Guarantees in Stochastic Systems} 

\author{Bai Xue$^1$, C.-H. Luke Ong$^2$, Dominik Wagner$^2$, and Peixin Wang$^3$}
\address{1. KLSS, ISCAS, Beijing, China; University of Chinese Academy of Sciences, Beijing, China. Email: xuebai@ios.ac.cn}
\address{2. College of Computing and Data Science, Nanyang Technological University, Singapore. Email: \{luke.ong,dominik.wagner\}@ntu.edu.sg}
\address{3. Software Engineering Institute, East China Normal University, China. Email: pxwang@sei.ecnu.edu.cn}



\begin{abstract}                
Providing finite-time probabilistic safety and reach-avoid guarantees is crucial for safety-critical stochastic systems. Existing state-of-the-art barrier  methods often rely on a restrictive boundedness assumption for auxiliary functions, limiting their applicability. This paper presents refined barrier conditions that remove this assumption. Specifically, we establish conditions for deriving upper bounds on finite-time safety probabilities in discrete-time systems and lower bounds on finite-time reach-avoid probabilities in continuous-time systems. This relaxation expands the class of verifiable systems, especially those with unbounded state spaces, and facilitates the use of advanced optimization techniques, such as semi-definite programming with polynomial functions. Numerical examples demonstrate the effectiveness of the approach.
\end{abstract}

\begin{keyword}
Stochastic Systems, Safety/Reach-avoid Probabilities, Barrier Conditions
\end{keyword}

\end{frontmatter}

\section{Introduction}
\label{sec:intro} 

Many real-world systems—such as autonomous vehicles and robotics—operate under strict time constraints \citep{alur2015principles,baier2008principles}. This leads to two key problems: finite-time safety (staying within a safe set) and reach-avoid verification (reaching a target while avoiding unsafe states) over a fixed horizon, both central to time-critical applications.

However, many systems are stochastic due to noise, disturbances, and interactions \citep{kumar2015stochastic}, making deterministic guarantees impractical. This motivates probabilistic notions such as finite-time safety and reach-avoid probabilities \citep{prajna2007framework,steinhardt2012finite,santoyo2021barrier,jagtap2018temporal,zhi2024unifying,chen2025construction,xue2024finite,xue2023new}, which quantify staying safe or reaching a target within a horizon. Since exact computation is intractable, prior work focuses on lower and upper bounds. Lower bounds provide conservative safety guarantees, while upper bounds capture performance limits and prevent overly optimistic assessments. Together, they define a probabilistic interval that supports risk-aware analysis and control design.

In this work, we propose new barrier-like conditions to bound finite-time safety and reach-avoid probabilities in both discrete- and continuous-time systems. For discrete time, we derive upper bounds on safety probabilities; for continuous time, lower bounds on reach-avoid probabilities. These results remove the restrictive bounded-function assumption and extend existing state-of-the-art methods. Our discrete-time conditions are inspired by the dynamic programming framework in \citep{chen2025construction}: reversing the inequalities yields valid upper-bound conditions, a technique also used in \citep{yu2023safe}. For continuous time, we introduce a new condition that avoids the boundedness requirement when lower-bounding reach-avoid probabilities. Numerical examples using semidefinite programming (SDP) demonstrate effectiveness.

The contributions of this work are as follows. We introduce new barrier-like conditions that bound finite-time safety and reach-avoid probabilities: upper bounds on safety probabilities for discrete-time systems and lower bounds on reach-avoid probabilities for continuous-time systems. By relaxing the  bounded-function requirement, our conditions extend verification to settings where such bounded functions do not exist. This relaxation enables efficient computational synthesis, e.g., SDP for polynomial templates, because polynomials (degree $\geq 1$) are unbounded on unbounded domains and were previously excluded by bounded-function assumptions. Our framework therefore permits SDP-based construction of unbounded polynomial barrier-like functions, making verification practical for a substantially broader class of systems.


\textbf{Related Work}.
Barrier certificates were introduced for deterministic safety and reachability verification \citep{prajna2004safety,prajna2007convex}, and later extended to stochastic systems for infinite-horizon analysis using Ville’s inequality and equation relaxations \citep{prajna2007framework,vzikelic2023learning,zhi2024unifying,xue2021reach,xue2022reach,yu2023safe,xue2024sufficient,chen2025construction}. Extensions also cover probabilistic programs and $\omega$-regular properties \citep{chakarov2013probabilistic,mciver2017new,kenyon2021supermartingales,abate2024stochastic,wang2025verifying}. Since this work considers finite-time verification, we next focus on bounded-horizon settings.


\subsubsection{Discrete-time Systems}
$c$-martingale barrier methods are widely used for stochastic discrete-time verification \citep{kushner1967,steinhardt2012finite}. They typically provide lower bounds on safety probabilities but not upper bounds. Extensions include temporal logic, controller synthesis, and neural $c$-martingales \citep{jagtap2018temporal,jagtap2020formal,santoyo2021barrier,mathiesen2022safety}, as well as results for difference inclusions and probabilistic programs \citep{ghanbarpour2025characterization,kura2019tail,chatterjee2024quantitative}. Recent work studies both bounds under strong invariance assumptions \citep{zhi2024unifying}, while \citep{xue2024finite} requires bounded barrier functions for upper bounds. In contrast, we remove this restriction and derive upper-bounding conditions.

\subsubsection{Continuous-time Systems}
Inspired by Kushner’s work \citep{kushner1967}, Steinhardt et al. \citep{steinhardt2012finite} extended barrier certificates to finite-horizon safety verification via $c$-martingales. Later refinements introduced state-dependent bounds \citep{santoyo2021barrier} and controlled variants \citep{wang2021safety}. These methods typically use Ville’s inequality to obtain lower bounds on safety probabilities. Another line of work \citep{xue2021reach,xue2022reach,yu2023safe} derives lower and upper bounds for reach-avoid probabilities via relaxed equations, with \citep{xue2023new} extending this to finite horizons. However, as in \citep{xue2024finite}, some finite-horizon conditions require bounded barrier functions, limiting common choices such as polynomials. We remove this boundedness requirement and propose more general barrier-like conditions.

\section{Preliminaries}  
\label{sec:pre}  
We introduce stochastic discrete- and continuous-time systems, the finite-time safety and reach-avoid problems, and related barrier-like conditions. We use the following notation: $\mathbb{R}$ is the set of real numbers; $\mathbb{N}$ the set of nonnegative integers; $\mathbb{N}_{\le k} := \{0,1,\dots,k\}$; $\Delta_1 \setminus \Delta_2$ denotes set difference, $\overline{\Delta_1}$ the closure of $\Delta_1$, and $\partial \Delta_1$ its boundary; and $1_A(\bm{x})$ is the indicator function of a set $A$.

\subsection{Problem Statement}

\subsubsection{Discrete-time Systems}

We consider stochastic discrete-time systems described by  
\begin{equation}
\label{system}
\bm{x}(l+1) = \bm{f}(\bm{x}(l),\bm{\theta}(l)), \quad \forall l\in \mathbb{N},
\end{equation}  
where \(\bm{x}(l)\in \mathbb{R}^n\) is the system state at time \(l\) and \(\bm{\theta}(l)\in \Theta\subseteq \mathbb{R}^m\) represents stochastic disturbances. The disturbance sequence \(\{\bm{\theta}(l)\}_{l\in \mathbb{N}}\) consists of i.i.d. (independent and identically distributed) random vectors defined on a probability space \((\Theta,\mathcal{F},\mathbb{P})\), with \(\bm{\theta}(l)\sim \mathbb{P}\) for all \(l\in \mathbb{N}\). The expectation \(\mathbb{E}[\cdot]\) is taken with respect to \(\mathbb{P}\).  

To define system trajectories, we first introduce the notion of a disturbance signal.  

\begin{definition}[Disturbance Signal]  
A \emph{disturbance signal} \(\pi\) is a sample path of the stochastic process \(\{\bm{\theta}(i) \colon \Theta \rightarrow \Theta, i \in \mathbb{N}\}\), defined on the canonical sample space \(\Theta^{\infty}\) with the product topology \(\mathcal{B}(\Theta^{\infty})\) and probability measure \(\mathbb{P}^{\infty}\). The associated expectation is denoted \(\mathbb{E}^{\infty}[\cdot]\).  
\end{definition}  

\begin{definition}[Trajectory]  
Given an initial state \(\bm{x}_0 \in \mathbb{R}^n\) and a disturbance signal \(\pi\), the trajectory of system \eqref{system} is the sequence  \(\bm{\phi}_{\pi}^{\bm{x}_0} \colon \mathbb{N} \rightarrow \mathbb{R}^n\) defined as  
\[
\bm{\phi}_{\pi}^{\bm{x}_0}(0) = \bm{x}_0,  \quad   \bm{\phi}_{\pi}^{\bm{x}_0}(l+1) = \bm{f}(\bm{\phi}_{\pi}^{\bm{x}_0}(l), \bm{\theta}(l)), \quad \forall l\in \mathbb{N}.
\]  
\end{definition}  

This paper addresses the verification problem of \emph{finite-time safety} over a horizon $\mathbb{N}_{\leq N}$, which seeks to upper bound the probability that the system remains within a designated safe set $\mathcal{X}$ throughout the horizon.

\begin{problem}
\label{safe}  
(Upper Bounds on Finite-Time Safety Probabilities) Given a finite time horizon $\mathbb{N}_{\leq N}$, a safe set $\mathcal{X}$, and an initial state $\bm{x}_0 \in \mathcal{X}$, determine $\epsilon_1 \in [0,1]$ such that $\mathbb{P}^{\infty}\Big(\forall k \in \mathbb{N}_{\leq N}. \bm{\phi}_{\pi}^{\bm{x}_0}(k) \in \mathcal{X} \Big) \leq \epsilon_1$. 
\end{problem}


\subsubsection{Continuous-time Systems}
We consider the continuous-time stochastic system, 
\begin{equation}
    \label{SDE}
    \begin{split}
d\bm{x}(t,\bm{w})=\bm{b}(\bm{x}(t,\bm{w}))dt+\bm{\sigma}(\bm{x}(t,\bm{w})) d\bm{W}(t,\bm{w}),
\end{split}
\end{equation}
where  $\bm{W}\colon \mathbb{R}\times \Omega\rightarrow \mathbb{R}^k$ is a standard $k$-dimensional Wiener process, and $\Omega$, equipped with the probability measure $\mathbb{P}$,  is the sample space $\bm{w}$ belongs to.  The expectation with respect to $\mathbb{P}$ is denoted by $\mathbb{E}[\,\cdot\,]$.

Besides, we have the following assumptions on 
$\bm{b}\colon \mathbb{R}^n\rightarrow \mathbb{R}^n$ and $\bm{\sigma}\colon \mathbb{R}^n \rightarrow \mathbb{R}^{n\times k}$.

\begin{assumption}[Local Lipschitz and Linear Growth]
\label{assump:sde_conditions}
The drift $\bm{b}:\mathbb{R}^n \to \mathbb{R}^n$ and diffusion $\bm{\sigma}:\mathbb{R}^n \to \mathbb{R}^{n \times k}$ of the SDE \eqref{SDE} satisfy:

\begin{enumerate}
    \item \textbf{Local Lipschitz continuity:} for any compact set $K \subset \mathbb{R}^n$, there exists $L_K>0$ such that
    \[
    \|\bm{b}(\bm{x}) - \bm{b}(\bm{y})\| + \|\bm{\sigma}(\bm{x}) - \bm{\sigma}(\bm{y})\|_F \le L_K \|\bm{x} - \bm{y}\|
    \]
    for any $\bm{x}, \bm{y} \in K$, where $\|\cdot\|_F$ denotes the Frobenius norm.
    
    \item \textbf{Linear growth:} there exists $C>0$ such that
    \[
    \|\bm{b}(\bm{x})\|^2 + \|\bm{\sigma}(\bm{x})\|_F^2 \le C (1 + \|\bm{x}\|^2)
    \]
    for any $\bm{x} \in \mathbb{R}^n$.
\end{enumerate}

Here, $\|\cdot\|$ denotes the Euclidean norm and $\|\cdot\|_F$ denotes the Frobenius norm.
\end{assumption}

Given an initial state $\bm{x}_0$, the SDE \eqref{SDE} has a unique strong solution over the time interval $[0,\infty)$. This solution is denoted as $\bm{X}_{\bm{x}_0}^{\bm{w}}\colon [0,\infty)\rightarrow \mathbb{R}^n$, which satisfies 
$\bm{X}_{\bm{x}_0}^{\bm{w}}(t)=\bm{x}_0+\int_{0}^t \bm{b}(\bm{X}_{\bm{x}_0}^{\bm{w}}(\tau))d \tau+\int_{0}^t \bm{\sigma}(\bm{X}_{\bm{x}_0}^{\bm{w}}(\tau)) d\bm{W}(\tau,\bm{w})$.

Let $v(t,\bm{x})$ be twice continuously differentiable over $\bm{x}$ and continuously differentiable over $t$, the infinitesimal generator underlying system \eqref{SDE} on the function $v(t,\bm{x})$, which represents the limit of the expected value of $v(t,\bm{X}_{\bm{x}_0}^{\bm{w}}(t))$ as $t$ approaches 0, is presented in Definition \ref{inf_generator}.
\begin{definition}
\label{inf_generator}
Given system \eqref{SDE},  the infinitesimal generator is the operator $\mathcal{L}$, which is defined to act on suitable functions $v: \mathbb{R}\times \mathbb{R}^n\rightarrow \mathbb{R}$ by
\begin{equation}
\label{infi}
\small
    \begin{split}
        \mathcal{L}v(t,\bm{x})=\lim_{\Delta t\rightarrow 0^+}\frac{\mathbb{E}[v(t+\Delta t,\bm{X}_{\bm{x}}^{\bm{w}}(t+\Delta t))]-v(t,\bm{x})}{\Delta t}.
    \end{split}
\end{equation}
The domain of $\mathcal{L}$ is by definition the set of functions $v$ for which the limit \eqref{infi} exists for all $\bm{x}\in \mathbb{R}^n$ and $t\in \mathbb{R}$. 
\end{definition}

The following proposition gives the formula for the infinitesimal generator $\mathcal{L}$. 

\begin{proposition}\citep{oksendal2013stochastic}
\label{prop:inf_generator}
Consider the system \eqref{SDE} under Assumption \ref{assump:sde_conditions}. The infinitesimal generator $\mathcal{L}$ acting on a function $v(t,\bm{x})$ is given by
$\mathcal{L}v(t,\bm{x}) = \frac{\partial v(t,\bm{x})}{\partial t}+ \nabla_{\bm{x}} v(t,\bm{x}) \cdot \bm{b}(\bm{x})+ \frac{1}{2} \mathrm{tr}\Big( \bm{\sigma}(\bm{x})^\top \bm{H}(v(t,\bm{x})) \bm{\sigma}(\bm{x}) \Big)$,
where $\nabla_{\bm{x}} v$ denotes the gradient and $\bm{H}(v)$ the Hessian of $v$ with respect to $\bm{x}$.
The domain of $\mathcal{L}$ includes all functions that are twice continuously differentiable over $\bm{x}\in \mathbb{R}^n$ and continuously differentiable over $t\in [0,T]$.

Furthermore, 
the process $\mathcal{L}v(t, \bm{X}_{\bm{x}_0}^{\bm{w}}(t))$ is integrable on $[0,T]$ and Dynkin's formula holds:
\begin{equation*}
\mathbb{E}[v(\tau, \bm{X}_{\bm{x}_0}^{\bm{w}}(\tau))] = v(0, \bm{x}_0) + \mathbb{E}\left[\int_0^\tau \mathcal{L}v(s, \bm{X}_{\bm{x}_0}^{\bm{w}}(s))  ds\right],
\end{equation*}
for stopping times $\tau \leq T$ a.s. \qed
\end{proposition}

We consider a verification problem over a finite time horizon $[0,T]$, with $T>0$. Given an open state set $\mathcal{X} \subseteq \mathbb{R}^n$ (\textbf{which may be unbounded}), the problem concerns \emph{finite-time reach-avoid} tasks: it aims to lower-bound the reach-avoid probability that the system reaches a closed target set $\mathcal{X}_r \subset \mathcal{X}$ while avoiding unsafe states $\mathbb{R}^n \setminus \mathcal{X}$, starting from an initial state in $\mathcal{X} \setminus \mathcal{X}_r$.

 \begin{problem}
\label{un_proII}
(Lower Bounds on Finite-time Reach-avoid Probabilities)Given a finite time horizon \([0,T]\), an open safe set \(\mathcal{X}\), a closed target set \(\mathcal{X}_r \subseteq \mathcal{X}\), and an initial state \(\bm{x}_0 \in \mathcal{X}\setminus \mathcal{X}_r\), determine \(\epsilon_2 \in [0,1]\) such that $\mathbb{P}_{\bm{x}_0}^{[0,T]}\geq \epsilon_2$, 
where $\mathbb{P}_{\bm{x}_0}^{[0,T]}=\mathbb{P}\left( \left\{\bm{w}\in \Omega \,\middle|\ 
\begin{aligned} &\exists t\in [0,T]. \bm{X}_{\bm{x}_0}^{\bm{w}}(t) \in \partial \mathcal{X}_r\\
&\wedge\forall \tau \in [0,t]. \bm{X}_{\bm{x}_0}^{\bm{w}}(\tau) \in \mathcal{X}  \end{aligned} 
\right\}
\right)$.

\end{problem}

This paper will propose new \emph{barrier-like conditions} for addressing Problems~\ref{safe}-\ref{un_proII}.




\subsection{Barrier-Like Conditions}  
\label{sub:dpf}  
We review existing state-of-the-art barrier-like conditions that provide upper bounds for safety probabilities in discrete-time systems and lower bounds for reach-avoid probabilities in continuous-time systems.

\subsubsection{2.2.1 Barrier-Like Conditions for Discrete-time Systems}
\begin{proposition}[Theorem 2, \citep{xue2024finite}] 
\label{pro_e7}  
Suppose \(\widetilde{\mathcal{X}}\subseteq \mathbb{R}^n\) and there exists \(v \colon \widetilde{\mathcal{X}} \to \mathbb{R}\) with \(\sup_{\bm{x}\in \widetilde{\mathcal{X}}} v(\bm{x}) \leq M\), \(\alpha \in [1,\infty)\), and \(\beta \in (1-\alpha,\infty)\) satisfying  
\begin{equation}
\label{pro_e71}
\begin{cases}
\beta + \alpha v(\bm{x}) \leq \mathbb{E}^{\infty}[v(\bm{\phi}_{\pi}^{\bm{x}}(1))], & \forall \bm{x} \in \mathcal{X},\\
v(\bm{x}) \leq 1, & \forall \bm{x} \in \widetilde{\mathcal{X}}\setminus \mathcal{X},
\end{cases}
\end{equation} 
where $\alpha$ and $\beta$ are user-defined scalar constants, then, for $\bm{x}_0\in \mathcal{X}$, $\mathbb{P}^{\infty}\Big( \forall k \in \mathbb{N}_{\leq N}. \bm{\phi}_{\pi}^{\bm{x}_0}(k) \in \mathcal{X}\Big)\leq $
\[
\begin{cases}
1-\frac{(\alpha^{N+1}v(\bm{x}_0)-M)(\alpha-1)+\beta(\alpha^{N+1}-1)}{(\alpha+\beta-1)(\alpha^{N+1}-1)}, & \alpha>1,\\
\frac{v(\bm{x}_0)-M}{\beta(N+1)}, & \alpha=1,
\end{cases}
\]  
where $\widetilde{\mathcal{X}}$ is a set satisfying $\widetilde{\Omega}\subset\widetilde{\mathcal{X}}$ and $\widetilde{\Omega}$ is the union of $\mathcal{X}$ and
all reachable states starting from $\mathcal{X}$ in one step, i.e.,  
\[\widetilde{\Omega}=\{\bm{x}'\in \mathbb{R}^n\mid \bm{x}'=\bm{f}(\bm{x},\bm{\theta}), \bm{x}\in \mathcal{X}, \bm{\theta}\in \Theta\}\cup \mathcal{X}.\]
\end{proposition}  


\subsubsection{2.2.2 Barrier-Like Conditions for Continuous-time Systems}

\begin{proposition}[Theorem 2, \citep{xue2023new}]
\label{coro5}
Assume the safe set $\mathcal{X}$ is bounded and open.  Suppose there exist a barrier function $v\colon [0,T]\times \mathbb{R}^n\rightarrow \mathbb{R}$ and a function $w\colon [0,T]\times \mathbb{R}^n \rightarrow \mathbb{R}$ with $\sup_{(t,\bm{x})\in [0,T]\times \overline{\mathcal{X}}}|w(t,\bm{x})|\leq M$ that are continuously differentiable over $t$ and twice continuously differentiable over $\bm{x}$, and $\beta \in \mathbb{R}$, satisfying 
   \begin{equation}
   \small
      \label{upper_bound_2}
        \begin{cases}
            \mathcal{L}v(t,\bm{x})\geq \alpha v(t,\bm{x})+\beta, & \forall \bm{x}\in \mathcal{X}\setminus \mathcal{X}_r,\forall t \in [0,T],\\
            \frac{\partial v(t,\bm{x})}{\partial t}\geq \alpha v(t,\bm{x})+\beta, & \forall \bm{x}\in \partial \mathcal{X}\cup \partial \mathcal{X}_r, \forall t\in [0,T],\\ 
            v(t,\bm{x})\leq 1+\frac{\partial w(t,\bm{x})}{\partial t}, &\forall \bm{x}\in \partial \mathcal{X}_r, \forall t\in [0,T],\\
            v(t,\bm{x})\leq \mathcal{L}w(t,\bm{x}), & \forall \bm{x}\in \mathcal{X}\setminus \mathcal{X}_r, \forall t \in [0,T],\\
            v(t,\bm{x})\leq \frac{\partial w(t,\bm{x})}{\partial t}, & \forall \bm{x}\in \partial \mathcal{X}, \forall t\in [0,T],
        \end{cases}
    \end{equation}
    where $\alpha$ is a user-defined scalar constant,
then
\begin{equation*}
\mathbb{P}_{\bm{x}_0}^{[0,T]}\geq 
\begin{cases}
\frac{(\frac{1}{\alpha}v(0,\bm{x}_0)+\frac{\beta}{\alpha^2})(e^{\alpha T}-1)-\frac{\beta}{\alpha}T}{T} -\frac{2M}{T}, & \text{if $\alpha\neq 0$}, \\
v(0,\bm{x}_0)+\frac{1}{2}\beta T-\frac{2M}{T}, & \text{if $\alpha=0$}
\end{cases}
\end{equation*}
for $\bm{x}_0\in \mathcal{X}\setminus \mathcal{X}_r$.    
\end{proposition}

\begin{proposition}[Corollary 3, \citep{xue2023new}]
\label{con_ex2}
    Suppose there exist twice continuously differentiable functions $v\colon \mathbb{R}^n\rightarrow \mathbb{R}$ and $w\colon \mathbb{R}^n \rightarrow \mathbb{R}$ with $\sup_{\bm{x}\in \overline{\mathcal{X}}}|w(\bm{x})|\leq M$, satisfying 
   \begin{equation}
      \label{upper_bound_3}
        \begin{cases}
            \mathcal{L}v(\bm{x})\geq \alpha v(\bm{x})+\beta, & \forall \bm{x}\in \mathcal{X}\setminus \mathcal{X}_r,\\
            0\geq \alpha v(\bm{x})+\beta, & \forall \bm{x}\in \partial \mathcal{X}\cup \partial \mathcal{X}_r\\ v(\bm{x})\leq 1 &\forall \bm{x}\in \partial \mathcal{X}_r,\\
            v(\bm{x})\leq \mathcal{L}w(\bm{x}), & \forall \bm{x}\in \mathcal{X}\setminus \mathcal{X}_r, \\
            v(\bm{x})\leq 0, & \forall \bm{x}\in \partial \mathcal{X},
        \end{cases}
    \end{equation}
    where $\alpha$ and $\beta$ are user-defined scalar constants,
then
\[\mathbb{P}_{\bm{x}_0}^{[0,T]}\geq
\begin{cases}
\frac{(\frac{1}{\alpha}v(\bm{x}_0)+\frac{\beta}{\alpha^2})(e^{\alpha T}-1)-\frac{\beta}{\alpha}T}{T} -\frac{2M}{T}, & \text{if $\alpha\neq 0$}, \\
v(\bm{x}_0)+\frac{1}{2}\beta T-\frac{2M}{T}, &\text{if $\alpha=0$}
\end{cases}
\]
for $\bm{x}_0\in \mathcal{X}\setminus \mathcal{X}_r$.    
\end{proposition}

It is observed that the condition \eqref{upper_bound_3} is independent of $T$. Therefore, if condition \eqref{upper_bound_3} holds, the lower bound in Proposition \eqref{con_ex2} holds for any $T>0$.

\section{Refined Barrier-like Conditions}
\label{sec:methods_discrete}
This section presents barrier-like conditions for upper bounds on finite-time safety (discrete-time systems) and lower bounds on reach–avoid probabilities (continuous-time systems), along with a theoretical comparison to existing results.

\subsection{Analysis on Barrier-like Conditions in Subsection \ref{sub:dpf}}

In this subsection, we analyze the barrier-like conditions introduced in Subsection \ref{sub:dpf}, highlighting their limitations and motivating our approach. 

The conditions in Propositions \ref{pro_e7}–\ref{con_ex2} exhibit several issues that may complicate or prevent practical computations:

(A) The conditions in Propositions~\ref{pro_e7}--\ref{con_ex2} require certain functions—specifically, $v(\bm{x})$ in discrete-time systems and $w(t,\bm{x})$, $w(\bm{x})$ in continuous-time systems—to be bounded. Because these functions and their boundedness conditions enter the computation, this requirement may complicate the overall procedure. Moreover, if the safe set $\mathcal{X}$ and/or $\widetilde{\mathcal{X}}$ is unbounded, the bounds may no longer be applicable for some functions, such as polynomials of degree one or higher in $\bm{x}$, since they cannot remain bounded over an unbounded domain. In such cases, $v(\bm{x})$ and $w(\bm{x})$ must be restricted to constants, while $w(t,\bm{x})$ can depend polynomially on $t$ only, with no dependence on $\bm{x}$.

Let us examine these conditions in more detail. We first consider the condition in Proposition~\ref{pro_e7}. If $v(\bm{x}) \equiv C$ for all $\bm{x} \in \mathbb{R}^n$, condition~\eqref{pro_e71} implies
$C \le \min\left\{\frac{\beta}{1-\alpha}, 1\right\}~\text{when }\alpha > 1$.  Thus, the upper bound for $\mathbb{P}^{\infty}\Big( \forall k \in \mathbb{N}_{\leq N}. \bm{\phi}_{\pi}^{\bm{x}_0}(k) \in \mathcal{X}\Big)$ will be
$\begin{cases}
\frac{(\alpha-1)(1-C)}{\alpha+\beta-1}, &\text{if~} \alpha>1,\\
0, & \text{if~}\alpha=0,
\end{cases}$
for any $\bm{x}_0\in \mathcal{X}$. Further, from $C \le \min\left\{\frac{\beta}{1-\alpha}, 1\right\}$ and $\beta>\alpha-1$, we have $\frac{(\alpha-1)(1-C)}{\alpha+\beta-1}\geq 1$, which is useless in practice.

Then, we examine the condition in Proposition \ref{coro5} with $w(t)$ instead of $w(t,\bm{x})$. In this case, since $v(t,\bm{x})$ is continuous over $\bm{x}$, the condition 
\eqref{upper_bound_2} is reduced to 
\begin{equation}
      \label{upper_bound_2_wt}
        \begin{cases}
            \mathcal{L}v(t,\bm{x})\geq v'(t,\bm{x}), & \forall \bm{x}\in \mathcal{X}\setminus \mathcal{X}_r,\forall t \in [0,T],\\
            \frac{\partial v(t,\bm{x})}{\partial t}\geq v'(t,\bm{x}), & \forall \bm{x}\in \partial \mathcal{X}\cup \partial \mathcal{X}_r, \forall t\in [0,T],\\
            v(t,\bm{x})\leq \frac{\partial w(t)}{\partial t}, & \forall \bm{x}\in \overline{\mathcal{X}\setminus \mathcal{X}_r}, \forall t\in [0,T],
        \end{cases}
    \end{equation}
    where $v'(t,\bm{x})=\alpha v(t,\bm{x})+\beta$.
    
In this context, the lower bounds in Proposition \ref{coro5} should also be less than or equal to 0. This claim can be justified via following the proof of Theorem 2 in \citep{xue2023new}. We also present the proof in Appendix. 

Finally, we examine the condition in Proposition \ref{con_ex2} with $w(\bm{x})\equiv C$ over $\bm{x}\in \mathbb{R}^n$. In this context, since $v(\bm{x})$ is continuous over $\bm{x}$, the condition \eqref{upper_bound_3} will be reduced to:
\begin{equation}
      \label{upper_bound_3_bw}
        \begin{cases}
            \mathcal{L}v(\bm{x})\geq \alpha v(\bm{x})+\beta, & \forall \bm{x}\in \mathcal{X}\setminus \mathcal{X}_r,\\
            0\geq \alpha v(\bm{x})+\beta, & \forall \bm{x}\in \partial \mathcal{X}\cup \partial \mathcal{X}_r\\ 
            v(\bm{x})\leq 0, & \forall \bm{x}\in \overline{\mathcal{X}\setminus \mathcal{X}_r}.
        \end{cases}
    \end{equation}
This is a time-independent counterpart of \eqref{upper_bound_2_wt}. Hence, by the same reasoning, the lower bounds in Proposition \ref{con_ex2} are also guaranteed to be non-positive.

\textbf{Combining the above analysis, we conclude that condition \eqref{upper_bound_2} with $w: \mathbb{R}\rightarrow \mathbb{R}$ and condition \eqref{upper_bound_3} with constant $w(\bm{x})\equiv C$ are not applicable to systems where the safe set $\mathcal{X}$ is unbounded. Thus, in the unbounded case, $w(t,\bm{x})$ and $w(\bm{x})$ cannot be  polynomials when solving \eqref{upper_bound_2} and \eqref{upper_bound_3}.}

\noindent (B) The conditions in Propositions \ref{coro5} and \ref{con_ex2} require more than one function, increasing computational burden.

In the following section, we present new conditions specifically designed to overcome these issues.

\subsection{Barrier-like Conditions for Discrete-time Systems}
This section present a new barrier-like condition for establishing upper bounds on finite-time reach–avoid probabilities in discrete-time systems. 

\begin{theorem}
    If there exist a function $v: \mathbb{R}^n \rightarrow \mathbb{R}$, $\alpha \geq 0$, and  $\beta\in \mathbb{R}$ satisfying 
    \begin{equation}
    \label{safe_upper}
        \begin{cases}
            v(\bm{x})\leq 1_{\mathbb{R}^n\setminus \mathcal{X}} (\bm{x}),\\
            \mathbb{E}_{\bm{\theta}}[v(\bm{f}(\bm{x},\bm{\theta}))] \geq \alpha v(\bm{x}) +\beta, &\forall \bm{x}\in \mathcal{X},\\
            v(\bm{x})\alpha^N+(\sum_{i=0}^{N-1} \alpha^i)\beta \leq 1, &\forall \bm{x}\in \mathbb{R}^n \setminus \mathcal{X},
        \end{cases}
    \end{equation}
    then, for $\bm{x}_0\in \mathcal{X}$, $\mathbb{P}^{\infty}\left(
\begin{aligned}
\forall  k\in \mathbb{N}_{\leq N}. \bm{\phi}_{\pi}^{\bm{x}_0}(k)\in \mathcal{X}
\end{aligned}
\right)\leq 
1-\big(v(\bm{x}_0)\alpha^N+\beta \sum_{i=0}^{N-1} \alpha^i \big)$.
\end{theorem}


\begin{pf}
   This conclusion can be justified by following the proof of Theorem 1 in \citep{chen2025construction}, with the inequality sign reversed. \qed
\end{pf}

\subsection{Barrier-like Conditions for Continuous-time stems}
This section presents new barrier-like conditions for establishing lower bounds on reach–avoid probabilities. 

The construction of the condition lies on an auxiliary stochastic process $\{\widetilde{\bm{X}}_{\bm{x}_0}^{\bm{w}}(t), t\in [0,\infty)\}$ for $\bm{x}_0\in \overline{\mathcal{X}\setminus \mathcal{X}_r}$ that is a stopped process corresponding to $\{\bm{X}_{\bm{x}_0}^{\bm{w}}(t), t\in [0,\infty)\}$ and the set $\overline{\mathcal{X}\setminus \mathcal{X}_r}$, i.e., 
\begin{equation}
\widetilde{\bm{X}}_{\bm{x}_0}^{\bm{w}}(t)=
\begin{cases}
\bm{X}_{\bm{x}_0}^{\bm{w}}(t) &\text{\rm~if~}t<\tau^{\bm{x}_0}(\bm{w}),\\
\bm{X}_{\bm{x}_0}^{\bm{w}}(\tau^{\bm{x}_0}(\bm{w})) & \text{\rm~if~}t\geq \tau^{\bm{x}_0}(\bm{w}),
\end{cases}
\end{equation}
where $\tau^{\bm{x}_0}(\bm{w})=\inf\{t\mid \bm{X}_{\bm{x}_0}^{\bm{w}}(t) \in \partial \mathcal{X}\cup \partial \mathcal{X}_r\}$ is the first time of hitting $\partial \mathcal{X}\cup \partial \mathcal{X}_r$. The stopped process $\widetilde{\bm{X}}_{\bm{x}_0}^{\bm{w}}(t)$ inherits the right continuity and strong Markovian property of $\bm{X}_{\bm{x}_0}^{\bm{w}}(t)$. Moreover, the infinitesimal generator corresponding to $\widetilde{\bm{X}}_{\bm{x}_0}^{\bm{w}}(t)$ is identical to the one corresponding to $\bm{X}_{\bm{x}_0}^{\bm{w}}(t)$ over $\mathcal{X}\setminus \mathcal{X}_r$, and is equal to zero on the boundary $\partial \mathcal{X}\cup \partial \mathcal{X}_r$ \citep{kushner1967}. That is, for $v(\bm{x},t)$ being twice continuously differentiable over $\bm{x}$ and continuously differentiable over $t$, 
\begin{equation}
\label{stop_l}
\begin{split}
\widetilde{\mathcal{L}}v(\bm{x},t)=\mathcal{L}v(\bm{x},t)
\end{split}
\end{equation} for $(t,\bm{x})\in [0,T] \times \mathcal{X}\setminus \mathcal{X}_r$ and $\widetilde{\mathcal{L}}v(\bm{x})=\frac{\partial v(t,\bm{x})}{\partial t}$ for $(t,\bm{x})\in [0,T] \times \mathcal{X}\setminus \mathcal{X}_r$. 

The reach-avoid probability of \eqref{SDE} on $[0,T]$ equals the probability that the auxiliary process hits $\partial \mathcal{X}_r$ at time $T$.
\begin{lemma}\citep{xue2023new}
\label{equiv}
    Given a time instant $T>0$ and $\bm{x}_0\in \mathcal{X}$, $
    \mathbb{P}_{\bm{x}_0}^{[0,T]}=\mathbb{P}(\{\bm{w}\mid \widetilde{\bm{X}}_{\bm{x}_0}^{\bm{w}}(T)\in \mathcal{X}_r\})=\mathbb{E}[1_{\mathcal{X}_r}(\widetilde{\bm{X}}_{\bm{x}_0}^{\bm{w}}(T))]$. Moreover, for any $0<T_1\leq T_2$, $\mathbb{P}(\{\bm{w}\mid \widetilde{\bm{X}}_{\bm{x}_0}^{\bm{w}}(T_1)\in \mathcal{X}_r\})\leq \mathbb{P}(\{\bm{w}\mid \widetilde{\bm{X}}_{\bm{x}_0}^{\bm{w}}(T_2)\in \mathcal{X}_r\})$.
\end{lemma}



\begin{theorem}
\label{thm_con}
    If there exists a barrier function $v\colon [0,T]\times \mathbb{R}^n\rightarrow \mathbb{R}$ that is continuously differentiable over $t$ and twice continuously differentiable over $\bm{x}$, satisfying 
 \begin{equation}
 \label{lower_bound_condition1}
 \small
     \begin{cases}
        \mathcal{L}v(t,\bm{x})\geq \alpha v(t,\bm{x})+\beta, & \forall (\bm{x},t)\in \mathcal{X}\setminus\mathcal{X}_r \times [0,T],\\
        \frac{\partial v(t,\bm{x})}{\partial t}+\alpha  \geq \alpha v(t,\bm{x}), & \forall (\bm{x},t)\in  \partial \mathcal{X}_r \times [0,T],\\
  \frac{\partial v(t,\bm{x})}{\partial t}\geq \alpha v(t,\bm{x})+\beta, &  \forall (\bm{x}, t)\in\partial \mathcal{X}\times [0,T],\\
  v(t,\bm{x})\leq 1, &  \forall (\bm{x},t) \in \overline{\mathcal{X}\setminus \mathcal{X}_r} \times [0,T],
    \end{cases}
\end{equation}   
then, 
$\mathbb{P}_{\bm{x}_0}^{[0,T]}\geq
    \begin{cases}
     \frac{v(0,\bm{x}_0)-1}{\beta T}+1, & \text{if $\alpha=0$}\\
      \frac{e^{\alpha T}(\alpha v(0,\bm{x}_0)+\beta)-(\alpha+\beta)}{(\alpha+\beta)(e^{\alpha T}-1)}, & \text{if  $\alpha>0$}
    \end{cases}$
for $\bm{x}_0\in \mathcal{X}\setminus \mathcal{X}_r$,  where $\alpha$ and $\beta$ are specified scalar values satisfying $\alpha\geq 0$ and $\alpha+\beta>0$.
\end{theorem}

The proof of Theorem \ref{thm_con} is presented in Appendix. Instead of employing a time-dependent barrier function $v\colon [0,T]\times \mathbb{R}^n \rightarrow \mathbb{R}$, a natural variant of Theorem \ref{thm_con} is to consider a time-independent barrier function $v\colon \mathbb{R}^n \rightarrow \mathbb{R}$.

\begin{corollary}
\label{coro3}
    Suppose there exists a function $v\colon \mathbb{R}^n\rightarrow \mathbb{R}$, which is twice continuously differentiable, satisfying 
      \begin{equation}
    \label{lower_bound_3}
        \begin{cases}
            \mathcal{L}v(\bm{x})\geq \alpha v(\bm{x})+\beta, &\forall \bm{x}\in \mathcal{X}\setminus \mathcal{X}_r, \\
           v(\bm{x})\leq 1, & \forall \bm{x}\in \overline{\mathcal{X}\setminus \mathcal{X}_r},\\
            \alpha v(\bm{x})+\beta \leq 0, & \forall \bm{x}\in \partial \mathcal{X},
        \end{cases}
    \end{equation} 
    where $\alpha$ and $\beta$ are specified scalar values satisfying $\alpha>0$ and $\alpha+\beta>0$, then, $\mathbb{P}_{\bm{x}_0}^{[0,T]}\geq \frac{e^{\alpha T}(\alpha v(\bm{x}_0)+\beta)-(\alpha+\beta)}{(\alpha+\beta)(e^{\alpha T}-1)}$
 for $\bm{x}_0\in \mathcal{X}\setminus \mathcal{X}_r$.    
\end{corollary}

The lower bounds in Theorem \ref{thm_con} and Corollary \ref{coro3} are monotonically increasing in $T$ and bounded by 1, provided $v(0,\bm{x}_0) \leq 1$ (or $v(\bm{x}_0) \leq 1$). This aligns with the monotonicity of the reach-avoid probability $\mathbb{P}_{\bm{x}_0}^{[0,T]}$. Notably, as condition \eqref{lower_bound_3} is $T$-independent, the bound in Corollary \ref{coro3} remains valid for all $T > 0$.


\section{Examples}
\label{sec:ex}
This section illustrates the applicability of the proposed safety and reach-avoid conditions through numerical examples. Rather than optimizing performance, we show that removing the bounded-function assumption enables rigorous verification for both discrete- and continuous-time stochastic systems, extending existing methods. All results are obtained using the SDP solver Mosek \citep{aps2019mosek}.

\subsection{Discrete-time Systems}

\begin{table}[h!]
    \centering
    \begin{adjustbox}{width=0.48\textwidth}
    \begin{tabular}{|c|c|c|c|c|c|c|c|c|}\hline

        
     \multicolumn{8}{|c|}{\eqref{safe_upper} with $\alpha=1.1$}\\\hline
   $d$ &2&4&6&8&10&12&14  \\\hline
        $\epsilon_1$&1.0000&0.9995& 0.9675& 0.8340& 0.7847& 0.7610&  0.7183\\\hline 
            \multicolumn{8}{|c|}{\eqref{safe_upper} with $\alpha=1$}\\\hline
   $d$ &2&4&6&8&10&12&14  \\\hline
        $\epsilon_1$&1.0000&0.9999& 0.9991& 0.9138& 0.7012& 0.6992&  0.5837\\\hline 
        
            \multicolumn{8}{|c|}{\eqref{safe_upper} with $\alpha=0.95$}\\\hline
   d &2&4&6&8&10&12&14  \\\hline
        $\epsilon_1$&1.0000&1.0000& 0.9996& 0.9990& 0.8001& 0.6846&  0.5141\\\hline  
    
    \end{tabular}
    \end{adjustbox}
      \caption{\small Upper bounds of the safety probability for Problem \ref{safe} in Example \ref{ex1} ($d$ denotes the degree of $v(\bm{x})$)}
     \label{tab:my_label_ex1}
\end{table}
\begin{table}[h]
    \centering
    \begin{adjustbox}{width=0.45\textwidth}
    \begin{tabular}{|c|c|c|c|c|c|c|}\hline
    \multicolumn{7}{|c|}{\eqref{safe_upper} with $\alpha=1.1$}\\\hline
   $d$ &2&4&6&8&10&12  \\\hline
        $\epsilon_1$&0.9997&0.9986& 0.9983& 0.9886&0.9103&0.8202\\\hline 
    
\multicolumn{7}{|c|}{\eqref{pro_e71} with $\alpha=1.1$, $\beta=0$, and $\tilde{\mathcal{X}}=\{(x,y)^{\top}\mid x^2+y^2\leq 30\}$}\\\hline
   $d$ &2&4&6&8&10&12  \\\hline
        $\epsilon_1$&1.0000&1.0000& 1.0000&1.0000 & 0.9790&0.9445\\\hline
        
    \end{tabular}
 \end{adjustbox}
      \caption{\small Upper bounds of the safety probability for Problem \ref{safe} in Example \ref{ex2} ($d$ denotes the degree of $v(\bm{x})$)}
     \label{tab:my_label_ex2}
\end{table}

\begin{example}
\label{ex1}
We consider the one-dimensional system
\begin{equation*}
x(l+1)=d(l)+x(l),
\end{equation*}
where $d\colon\mathbb{N}\rightarrow [-1,1]$, the initial state is $\bm{x}_0=-0.9$, the safe set is $\mathcal{X}=\{x\mid x\leq 0\}$, and the horizon is $N=10$. The disturbance distribution is uniform on $[-1,1]$. 

Since $\widetilde{\mathcal{X}}$ is unbounded, non-constant polynomials cannot satisfy the boundedness required by \eqref{pro_e71}. Conversely, our condition \eqref{safe_upper} remains applicable, yielding a tight upper bound of 0.5141 (degree 14, $\alpha=0.95$) compared to the Monte Carlo estimate of 0.4794 (Table~\ref{tab:my_label_ex1}).


\end{example}
\begin{table}[h]
    \centering
    \begin{adjustbox}{width=0.48\textwidth}
    \begin{tabular}{|c|c|c|c|c|c|c|c|}\hline
    \multicolumn{8}{|c|}{\eqref{lower_bound_condition1} with $\alpha=0.1$ and $\beta=0$}\\\hline
  $(d_x,d_t)$ &(8,1)&(10,1)&(12,1)&(14,1)&(16,1)&(18,1)&(20,1)  \\\hline
        $\epsilon_1$& -&-&0.7297& 0.7646&0.7873 & 0.7921&0.8033\\\hline

\multicolumn{8}{|c|}{\eqref{lower_bound_condition1} with $\alpha=0.1$ and $\beta=-0.09$}\\\hline
  $(d_x,d_t)$ &(8,1)&(10,1)&(12,1)&(14,1)&(16,1)&(18,1)&(20,1)  \\\hline
        $\epsilon_1$& 0.3129&0.6163&0.7807& 0.8244&0.8380 & 0.8472&0.8583\\\hline

        \multicolumn{8}{|c|}{\eqref{lower_bound_condition1} with $\alpha=0.1$ and $\beta=-0.098$}\\\hline
  $(d_x,d_t)$ &(8,1)&(10,1)&(12,1)&(14,1)&(16,1)&(18,1)&(20,1)  \\\hline
        $\epsilon_1$& 0.5355&0.7201&0.8145& 0.8423&0.8774
 & 0.9736&0.9950\\\hline
        
    \multicolumn{8}{|c|}{\eqref{lower_bound_3} with $\alpha=0.1$  and $\beta=0$}\\\hline
   $(d_x,d_t)$ &(8,0)&(10,0)&(12,0)&(14,0)&(16,0)&(18,0)&(20,0)  \\\hline
        $\epsilon_1$&0.0000&-&0.6303&0.7572& 0.7690& 0.7846& 0.7962\\\hline 
    
\multicolumn{8}{|c|}{\eqref{lower_bound_3} with $\alpha=0.1$ and $\beta=-0.09$}\\\hline
  $(d_x,d_t)$ &(8,0)&(10,0)&(12,0)&(14,0)&(16,0)&(18,0)&(20,0)  \\\hline
        $\epsilon_1$& 0.1247&0.5165&0.6786& 0.7937&0.8249 & 0.8355&0.8407\\\hline

        \multicolumn{8}{|c|}{\eqref{lower_bound_3} with $\alpha=0.1$ and $\beta=-0.098$}\\\hline
  $(d_x,d_t)$ &(8,0)&(10,0)&(12,0)&(14,0)&(16,0)&(18,0)&(20,0)  \\\hline
        $\epsilon_1$& 0.2389&0.59075&0.7555& 0.8219&0.8396 & 0.8465&0.8503\\\hline

    \end{tabular}
 \end{adjustbox}
      \caption{\small Lower bounds of the reach-avoid probability $\mathbb{P}_{\bm{x}_0}^{[0,T]}$ in Example \ref{ex3} with $\mathcal{X}=\{x\mid x<1\}$ ($d_x$ and $d_t$ are the degrees of $v(t,\bm{x})$ in $\bm{x}$ and $t$, respectively.)}
     \label{tab:my_label_ex3}
\end{table}

 \begin{table}[h]
    \centering
    \begin{adjustbox}{width=0.48\textwidth}
    \begin{tabular}{|c|c|c|c|c|c|c|c|}\hline

    \multicolumn{8}{|c|}{\eqref{lower_bound_condition1} with $\alpha=0.1$ and $\beta=0$}\\\hline
  $(d_x,d_t)$ &(8,1)&(10,1)&(12,1)&(14,1)&(16,1)&(18,1)&(20,1)  \\\hline
        $\epsilon_1$& 0.2235&0.4729&0.6713& 0.7620&0.7893 & 0.8010&0.8096\\\hline

\multicolumn{8}{|c|}{\eqref{lower_bound_condition1} with $\alpha=0.1$ and $\beta=-0.09$}\\\hline
  $(d_x,d_t)$ &(8,1)&(10,1)&(12,1)&(14,1)&(16,1)&(18,1)&(20,1)  \\\hline
        $\epsilon_1$& 0.4194&0.6520&0.7661& 0.8217&0.8397 & 0.8606&0.8672\\\hline

        \multicolumn{8}{|c|}{\eqref{lower_bound_condition1} with $\alpha=0.1$ and $\beta=-0.098$}\\\hline
  $(d_x,d_t)$ &(8,1)&(10,1)&(12,1)&(14,1)&(16,1)&(18,1)&(20,1)  \\\hline
        $\epsilon_1$& 0.9824&0.9868&0.9899& 0.9865&0.9833 & 0.9804&0.9881\\\hline
        
    \multicolumn{8}{|c|}{\eqref{lower_bound_3} with $\alpha=0.1$ and $\beta=0$}\\\hline
   $(d_x,d_t)$ &(8,0)&(10,0)&(12,0)&(14,0)&(16,0)&(18,0)&(20,0)  \\\hline
        $\epsilon_1$&0.1222&-&0.6540&0.7555& 0.7767& 0.7942& 0.8052\\\hline 
    
\multicolumn{8}{|c|}{\eqref{lower_bound_3} with $\alpha=0.1$ and $\beta=-0.09$}\\\hline
  $(d_x,d_t)$ &(8,0)&(10,0)&(12,0)&(14,0)&(16,0)&(18,1)&(20,0)  \\\hline
        $\epsilon_1$& 0.3599&0.6111&0.7320& 0.8020&0.8358 & 0.8459&0.8616\\\hline

        \multicolumn{8}{|c|}{\eqref{lower_bound_3} with $\alpha=0.1$ and $\beta=-0.098$}\\\hline
  $(d_x,d_t)$ &(8,0)&(10,0)&(12,0)&(14,0)&(16,0)&(18,1)&(20,0)  \\\hline
        $\epsilon_1$& 0.4809&0.7079&0.8111& 0.8473&0.8551 & 0.8814&0.8815\\\hline

  \end{tabular}
 \end{adjustbox}
      \caption{\small Lower  bounds of the reach-avoid probability $\mathbb{P}_{\bm{x}_0}^{[0,T]}$ in Example \ref{ex3} with $\mathcal{X}=\{x\mid x^2<100\}$ ($d_x$ and $d_t$ are the degrees of $v(t,\bm{x})$ in $\bm{x}$ and $t$, respectively.)}
     \label{tab:my_label_ex31}
\end{table}
\begin{example}
\label{ex2}
We study the Lotka--Volterra population model:
\begin{equation}
\label{volterra}
    \begin{cases}
x(l+1)=rx(l)-ay(l)x(l),\\
y(l+1)=sy(l)+acy(l)x(l),
\end{cases}
\end{equation}
with $r=0.5$, $a=1$, $s=-0.5+d(l)$, $d\colon\mathbb{N}\rightarrow [-1,1]$, $c=1$, and the safe set $\mathcal{X}=\{\bm{x}\mid x^2+y^2\leq 4\}$. The initial state is $\bm{x}_0=(-0.8,-0.6)^{\top}$ and the horizon is $N=50$. 

Table~\ref{tab:my_label_ex2} compares the results obtained from \eqref{safe_upper} and \eqref{pro_e71}, both of which successfully certify nontrivial probability bounds. In \eqref{pro_e71}, the value of $\beta$ is specified a priori, because the resulting upper bound  when $\alpha>1$ is nonlinear in $\beta$. A Monte Carlo simulation with $10^6$ trajectories provides a reference probability of $0.8141$. It is observed that condition \eqref{safe_upper} with $\alpha=1.1$ and polynomial degree $12$ yields an upper bound of $0.8202$, which is close to the Monte Carlo estimate of $0.8141$.
\end{example}
\subsection{Continuous-time Systems}
\begin{example}
\label{ex3}
    Consider the one-dimensional system:
 \[
dX(t,w)=b(X(t,w))\,dt+\sigma(X(t,w))\,dW(t,w), 
\]
with $b(x)=-x$ and $\sigma(x)=\tfrac{\sqrt{2}}{2}x$. This diffusion model is often used for population dynamics under random environmental fluctuations. We set the initial state to $\bm{x}_0=-0.8$, the horizon to $T=100$, and the target set $\mathcal{X}_r=\{x\in\mathbb{R}\mid 100x^2-1\leq 0\}$. Two safe sets are considered: $\mathcal{X}=\{x\mid x<1\}$ \text{and} $\mathcal{X}=\{x\mid x^2<100\}$.

Conditions \eqref{lower_bound_condition1} and \eqref{lower_bound_3} yield computable lower bounds (Tables~\ref{tab:my_label_ex3}--\ref{tab:my_label_ex31}) consistent with the Monte Carlo estimate based on $10^6$ trajectories, which is $1$. Unlike existing bounded-function approaches (i.e., \eqref{upper_bound_2} and \eqref{upper_bound_3}) that fail on unbounded $\widetilde{\mathcal{X}}$, our framework remains valid and produces tight lower bounds.

\end{example}



       



The case studies highlight four aspects of the proposed method: (1) \textbf{Generality}: removing the bounded-function restriction enables analysis of systems and safe sets beyond existing approaches (e.g., Examples~\ref{ex1}, \ref{ex3}); (2) \textbf{Soundness}: all certificates yield rigorous probability bounds consistent with Monte Carlo estimates; (3) \textbf{Efficiency}:  higher polynomial degrees improve tightness at increased computational cost; (4) \textbf{Parameter sensitivity}:  the tightness depends on $\alpha$ in \eqref{safe_upper} and $(\alpha,\beta)$ in \eqref{lower_bound_condition1}–\eqref{lower_bound_3}, motivating future work on automatic tuning.


\section{Conclusion}
\label{sec:con}
This paper developed barrier-like conditions for upper-bounding finite-time safety probabilities in discrete-time systems and lower-bounding reach-avoid probabilities in continuous-time systems. By removing the bounded-function requirement, the approach extends existing state-of-the-art methods and applies to a broader class of systems. Numerical examples demonstrate its effectiveness.

\textbf{Acknowledgments.} This work was partially supported by the National Research Foundation, Singapore, under its RSS Scheme (NRF-RSS2022-009) and the CAS Pioneer Hundred Talents Program.

\bibliography{ifacconf}             
\section*{Appendix}
\label{sec:app}

\subsection{Proofs of Theorem~2 and Proposition~5}

\begin{proposition}
     Suppose there exist a barrier function $v(t,\bm{x})\colon [0,T]\times \mathbb{R}^n\rightarrow \mathbb{R}$ and a function $w(t)\colon [0,T] \rightarrow \mathbb{R}$ with $\sup_{t\in [0,T]}|w(t)|\leq M$ that are continuously differentiable over $t$ and twice continuously differentiable over $\bm{x}$, and $\beta\in \mathbb{R}$ satisfying \eqref{upper_bound_2_wt}, where $\alpha$ is a user-defined scalar value, then, 
     \begin{equation*}
     \begin{cases}
\frac{(\frac{1}{\alpha}v(0,\bm{x}_0)+\frac{\beta}{\alpha^2})(e^{\alpha T}-1)-\frac{\beta}{\alpha}T}{T} -\frac{2M}{T}\leq 0, & \text{if $\alpha\neq 0$}, \\
v(0,\bm{x}_0)+\frac{1}{2}\beta T-\frac{2M}{T} \leq 0, & \text{if $\alpha=0$}
\end{cases}
\end{equation*}
for $\bm{x}_0\in \mathcal{X}\setminus \mathcal{X}_r$. 
     
\end{proposition}
\begin{pf}
From \eqref{upper_bound_2_wt}, we have
 \begin{equation*}
        \begin{cases}
            \widetilde{\mathcal{L}}v(t,\bm{x})\geq \alpha v(t,\bm{x})+\beta, & \forall \bm{x}\in \overline{\mathcal{X}\setminus \mathcal{X}_r},\forall t \in [0,T],\\
            v(t,\bm{x})\leq \frac{\partial w(t)}{\partial t},&  \forall \bm{x}\in \overline{\mathcal{X}\setminus \mathcal{X}_r}, \forall t \in [0,T],
        \end{cases}
    \end{equation*}
    where $\widetilde{\mathcal{L}}v(t,\bm{x})$ is defined in \eqref{stop_l}.

When $\alpha\neq 0$, according to $\widetilde{\mathcal{L}}v(t,\bm{x})\geq \alpha v(t,\bm{x})+\beta, \forall \bm{x}\in \overline{\mathcal{X}\setminus \mathcal{X}_r},\forall t \in [0,T]$, we have, for $t\in [0,T]$, that 
\begin{equation*}
\begin{split}
    \mathbb{E}[v(t,\widetilde{\bm{X}}_{\bm{x}_0}^{\bm{w}}(t))]\geq e^{\alpha t}v(0,\bm{x}_0)+\frac{\beta}{\alpha}(e^{\alpha t}-1).
    \end{split}
    \end{equation*}

Further, from  $v(t,\bm{x})\leq \frac{\partial w(t)}{\partial t}, \forall t \in [0,T], \forall \bm{x}\in \overline{\mathcal{X}\setminus \mathcal{X}_r}$, we have 
\begin{equation*}
\begin{split}
&0\geq \frac{\int_{0}^T \mathbb{E}[v(t,\widetilde{\bm{X}}_{\bm{x}_0}^{\bm{w}}(t))]dt}{T}-\frac{\mathbb{E}[w(T)]-w(0)}{T}\\
&\geq \frac{\int_{0}^T e^{\alpha t}v(0,\bm{x}_0)+\frac{\beta}{\alpha}(e^{\alpha t}-1) dt}{T}- \frac{\mathbb{E}[w(T)]-w(0)}{T}\\
&\geq \frac{(\frac{1}{\alpha}e^{\alpha t}v(0,\bm{x}_0)+\frac{\beta}{\alpha^2}e^{\alpha t}-\frac{\beta}{\alpha}t)\mid_{0}^T}{T}-\frac{2M}{T}\\
&= \frac{(\frac{1}{\alpha}v(0,\bm{x}_0)+\frac{\beta}{\alpha^2})(e^{\alpha T}-1)-\frac{\beta}{\alpha}T}{T} -\frac{2M}{T}.
  \end{split}
\end{equation*}

The conclusion for $\alpha=0$ can be obtained by following the above procedure. The proof is completed. 
\end{pf}

\textbf{The proof of Theorem \ref{thm_con}:}

\begin{pf}
    According to \eqref{lower_bound_condition1}, we have 
    \begin{equation*}
\widetilde{\mathcal{L}}v(t,\bm{x})+(\alpha+\beta)1_{\partial \mathcal{X}_r}(\bm{x})\geq  \alpha v(t,\bm{x})+\beta, \forall \bm{x}\in \overline{\mathcal{X}}, \forall t\in [0,T],
\end{equation*} 
where \[\widetilde{\mathcal{L}}v(t,\bm{x})=\begin{cases}
    \mathcal{L}v(t,\bm{x}), & \text{~if~} \bm{x}\in  \mathcal{X}\setminus \mathcal{X}_r,  t\in [0,T],\\
    \frac{\partial v(t,\bm{x})}{\partial t},&\text{~if~} \bm{x}\in \partial \mathcal{X} \cup \partial \mathcal{X}_r, t\in [0,T].
\end{cases}
\]

Consequently, 
\[
\begin{split}
\mathbb{E}&[v(T,\widetilde{\bm{X}}_{\bm{x}_0}^{\bm{w}}(T))]\geq \int_{0}^T \alpha \mathbb{E}[v(t,\widetilde{\bm{X}}_{\bm{x}_0}^{\bm{w}}(t))]dt+v(0,\bm{x}_0)\\
&~~+\int_{0}^T \beta dt-\int_{0}^T (\alpha+\beta)\mathbb{E}[1_{\mathcal{X}_r}(\widetilde{\bm{X}}_{\bm{x}_0}^{\bm{w}}(t))]dt, \forall \bm{x}_0\in \mathcal{X}.
\end{split}
\]
Taking $\overline{v}(t,\bm{x})=-v(t,\bm{x})$ over $\bm{x}\in \overline{\mathcal{X}}$ and $t\in [0,T]$, we have 
\[
\begin{split}
\mathbb{E}&[\overline{v}(T,\widetilde{\bm{X}}_{\bm{x}_0}^{\bm{w}}(T))]\leq \int_{0}^T \alpha \mathbb{E}[\overline{v}(t,\widetilde{\bm{X}}_{\bm{x}_0}^{\bm{w}}(t))]dt+\overline{v}(0,\bm{x}_0)\\
&~~-\int_{0}^T \beta dt+\int_{0}^T (\alpha+\beta)\mathbb{E}[1_{\mathcal{X}_r}(\widetilde{\bm{X}}_{\bm{x}_0}^{\bm{w}}(t))]dt.
\end{split}
\]

According to Gr\"onwall inequality in the integral form, if $\alpha>0$, we further have 
\[
\begin{split}
&\mathbb{E}[\overline{v}(T,\widetilde{\bm{X}}_{\bm{x}_0}^{\bm{w}}(T))]\leq \nu(T)+\int_{0}^T \nu(s) \alpha  e^{\alpha (T-s)}ds\\
&= \overline{v}(0,\bm{x}_0)+\int_{0}^T \overline{v}(0,\bm{x}_0)\alpha e^{\alpha(T-s)} ds -\beta T\\
&-\int_{0}^{T} \beta s\alpha e^{\alpha(T-s)}ds+ (\alpha+\beta)\int_{0}^T \mathbb{E}[1_{\mathcal{X}_r}(\widetilde{\bm{X}}_{\bm{x}_0}^{\bm{w}}(s))] ds \\
&~~~~~~~~~~+\alpha (\alpha+\beta)\int_{0}^T \int_{0}^s \mathbb{E}[1_{\mathcal{X}_r}(\widetilde{\bm{X}}_{\bm{x}_0}^{\bm{w}}(t))] dt e^{\alpha (T-s)} ds    \\
&\leq \overline{v}(0,\bm{x}_0)e^{\alpha T}+\frac{\beta}{\alpha}-\frac{\beta}{\alpha}e^{\alpha T} \\
&~~~~~~~~~~~~~+(\alpha+\beta)e^{\alpha T}\mathbb{P}(\widetilde{\bm{X}}_{\bm{x}_0}^{\bm{w}}(T) \in \mathcal{X}_r)(-\frac{1}{\alpha}e^{-\alpha T}+\frac{1}{\alpha}) 
\end{split}
\]
where $\nu(s)=\overline{v}(0,\bm{x}_0)+\int_{0}^s (\alpha+\beta)\mathbb{E}[1_{ \mathcal{X}_r}(\widetilde{\bm{X}}_{\bm{x}_0}^{\bm{w}}(t))]dt-\int_{0}^s \beta dt$. The last inequality is obtained from Lemma \ref{equiv}.

Thus, 
\[
\begin{split}
&-1\leq \mathbb{E}[\overline{v}(T,\widetilde{\bm{X}}_{\bm{x}_0}^{\bm{w}}(T))]\leq \overline{v}(0,\bm{x}_0)e^{\alpha T}+\frac{\beta}{\alpha}-\frac{\beta}{\alpha}e^{\alpha T} \\
&~~~~~~~~~~~~~+(\alpha+\beta)e^{\alpha T}\mathbb{P}(\widetilde{\bm{X}}_{\bm{x}_0}^{\bm{w}}(T) \in  \mathcal{X}_r)(-\frac{1}{\alpha}e^{-\alpha T}+\frac{1}{\alpha}) 
\end{split}
\]
After rearrangement, we conclude that $\mathbb{P}(\widetilde{\bm{X}}_{\bm{x}_0}^{\bm{w}}(T)\in \mathcal{X}_r)\geq \frac{e^{\alpha T}(v(0,\bm{x}_0)+\frac{\beta}{\alpha})-\frac{\beta}{\alpha}-1}{(1+\frac{\beta}{\alpha})(e^{\alpha T}-1)}$ if $\alpha>0$. Furthermore, according to Lemma \ref{equiv}, $\mathbb{P}_{\bm{x}_0}^{[0,T]}\geq \frac{e^{\alpha T}(v(0,\bm{x}_0)+\frac{\beta}{\alpha})-\frac{\beta}{\alpha}-1}{(1+\frac{\beta}{\alpha})(e^{\alpha T}-1)}$  if $\alpha>0$.

When $\alpha=0$, we can obtain $\frac{v(0,\bm{x}_0)-1}{\beta T}+1$ via following the above arguments. \qed
\end{pf} 

\subsection{Clarification on Stopped-Process Generator and Boundary Regularity}
\begin{remark}
\begin{enumerate}
    \item \textbf{Stopped-process generator:} For any point $\bm{x} \in \partial \mathcal{X} \cup \partial \mathcal{X}_r$, the auxiliary process is absorbed by construction, i.e., $\widetilde{\bm{X}}_{\bm{x}}^{\bm{w}}(t) \equiv \bm{x}$ for $t \ge 0$. Hence, for any function $v(t,\bm{x})$ that is continuously differentiable in $t$ and twice continuously differentiable in $\bm{x}$,
    \[
    \widetilde{\mathcal{L}}v(t,\bm{x})
    =
    \lim_{\Delta t \to 0^+}
    \frac{v(t+\Delta t,\bm{x})-v(t,\bm{x})}{\Delta t}
    =
    \frac{\partial v}{\partial t}(t,\bm{x}).
    \]
    Thus, after stopping, the spatial derivative terms vanish, and the boundary behavior is purely temporal. This pointwise identity does not depend on geometric regularity of the boundary. We note that the statement $\widetilde{\mathcal{L}}v=0$ applies only to time-independent functions $v$; Theorem~2 uses the time-dependent form above.

    \item \textbf{Dynkin’s formula:} Under Assumption~1, the SDE admits a unique strong solution with continuous sample paths. Dynkin’s formula is applied in stopped form up to $\tau$, and its validity requires only the usual integrability condition
    \[
    \mathbb{E}\!\left[\int_0^T |\widetilde{\mathcal{L}}v(t,\widetilde{\bm{X}}_{\bm{x}_0}^{\bm{w}}(t))|\,dt\right] < \infty.
    \]
    In our setting, $v$ is chosen as a polynomial; together with the linear growth assumption on the SDE coefficients, this implies that $\widetilde{\mathcal{L}}v$ has at most polynomial growth and the above expectation is finite on any finite horizon $[0,T]$.

    \item \textbf{Boundary regularity:} Since sample paths are almost surely continuous, the process reaches $\partial \mathcal{X} \cup \partial \mathcal{X}_r$ at the first hitting time $\tau$. Therefore, no additional boundary regularity assumptions are required for the stopping time to be well defined or for the stopped-process generator characterization to hold.
\end{enumerate}
\end{remark}

\subsection{Clarification on SOS Enforcement and Domain Assumptions}

\begin{remark}
    Although $\partial\mathcal{X}$ and $\partial\mathcal{X}_r$ are equality-defined sets, the associated boundary constraints can be handled naturally in Sum-of-Squares (SOS) optimization using the S-procedure or Putinar’s Positivstellensatz. For example, a condition
\[
v(t,x)\ge c \quad \text{on } \{x\mid g(x)=0\}
\]
can be enforced by requiring
\[
v(t,x)-c+\lambda(x)g(x)
\]
to be SOS for some polynomial multiplier $\lambda(x)$. This is standard in SOS-based verification and allows equality constraints on boundaries to be enforced directly through polynomial feasibility conditions.

Regarding the condition $v\le 1$ on $\mathcal{X}\setminus\mathcal{X}_r$, this is only a one-sided normalization condition needed for the probabilistic interpretation of the certificate. It is fundamentally different from the boundedness assumptions in prior work, which require auxiliary functions themselves to remain bounded over the domain. In contrast, $v\le 1$ is compatible with unbounded domains and polynomial templates; for instance, one may choose $v(x)=1-p(x)$ with $p(x)\ge 0$. Similar observations apply to the condition
\[
\mathcal{L}v(t,x)\ge \alpha v(t,x)+\beta,
\]
which imposes only a differential inequality rather than boundedness of $v$.
\end{remark}

\begin{remark}

For numerical enforcement in the SOS/SDP setting for condition \eqref{safe_upper}, constraints over regions such as $\mathbb{R}^n \setminus \mathcal{X}$ are handled via polynomial multiplier relaxations (S-procedure). In particular, a constraint of the form $v(x)\le 1$ on a semialgebraic set defined by $p(x)\ge 0$ is enforced by requiring $1-v(x) - \lambda(x) p(x)$ to be SOS for some polynomial multiplier $\lambda(x) \ge 0$. This provides a tractable sufficient condition for verification without introducing artificial compact subsets.

Importantly, the theoretical guarantee in Theorem 1 does not rely on any bounded outer approximation of $\mathbb{R}^n \setminus \mathcal{X}$. The auxiliary set $\tilde{\mathcal{X}}$ in Table 2 is used only for the baseline method, which requires boundedness for its original formulation, whereas our proposed condition is directly applicable on the unbounded domain.

We acknowledge that domain size and geometry can influence the numerical conservatism (i.e., the polynomial degree required to find a feasible solution). However, this is a computational trade-off regarding the tightness of the bound, rather than a limitation of the theoretical guarantee. Once a certificate is found, the resulting probabilistic bound is mathematically valid for the full unbounded domain.
\end{remark}

                                                   








\end{document}